\documentclass[journal,twoside,web]{ieeecolor}

\usepackage{generic}
\usepackage{cite}
\usepackage{amsmath,amssymb,amsfonts}
\usepackage{algorithmic}
\usepackage{graphicx}
\usepackage{textcomp}

\usepackage{graphicx}
\usepackage{subfigure}
\usepackage{epsfig} 
\usepackage{cite}
\usepackage{lcsys}

\usepackage{diagrams}
\usepackage{graphicx}          
\usepackage{mathrsfs}
\usepackage{graphicx}
\usepackage{subfigure}
\usepackage{url}
\usepackage{amsmath}
\usepackage{color}
\usepackage{array}
\usepackage{dsfont}
\usepackage{amsfonts}
\usepackage{color}
\usepackage{flushend}
\usepackage{amssymb}
\usepackage{multicol}
\usepackage{cite}

\newtheorem{theorem}{Theorem}
\newtheorem{proposition}{Proposition}

\newtheorem{definition}{Definition}
\newtheorem{lemma}{Lemma}

\newtheorem{example}{Example}

\def\boxspan{\mathit{span}}
\def\params{{\mathsf{q}}}

\newcommand{\R}{{\mathbb{R}}}

\newcommand{\N}{{\mathbb{N}}}

\begin{document}

\def\BibTeX{{\rm B\kern-.05em{\sc i\kern-.025em b}\kern-.08em
    T\kern-.1667em\lower.7ex\hbox{E}\kern-.125emX}}
\markboth{\journalname, VOL. XX, NO. XX, XXXX 2017}
{Author \MakeLowercase{\textit{et al.}}: Preparation of Papers for IEEE Control Systems Letters (August 2022)}

\title{Abstraction-Based Verification of Approximate Pre-Opacity  for
	Control Systems}

\author{Junyao Hou$^{1}$, Siyuan Liu$^{2,3}$, Xiang Yin$^{1}$ and Majid Zamani$^{3,4}$
	\thanks{This work was supported  by  the National Natural Science Foundation of China (62061136004,62173226, 61833012),  the German Research
		Foundation under Grant ZA 873/7-1, and the National Science Foundation under Grant ECCS-2015403.}%
	\thanks{$^1$Department of Automation and Key Laboratory of System Control and Information Processing,	Shanghai Jiao Tong University, China. 
		$^2$Department of Electrical and Computer Engineering, Technical University of Munich, Germany.
		$^3$Department of Computer Science, LMU Munich, Germany.
		$^4$Department of Computer Science, University of Colorado Boulder, USA. 
		The first two authors contributed equally to this work. Corresponding author X.~Yin. E-mail: {\tt\small  \{yinxiang\}@sjtu.edu.cn}.}
}

\maketitle
\thispagestyle{empty}

	\begin{abstract}
	In this paper, we consider the problem of verifying pre-opacity for discrete-time control systems.
	Pre-opacity is an important information-flow security property that secures the intention of a system to execute some secret behaviors in the future. 
	Existing works on pre-opacity only consider  non-metric discrete systems, where it is assumed that intruders can distinguish different output behaviors precisely. 
	However, for continuous-space control systems whose output sets are equipped with metrics (which is the case for most real-world applications), it is too restrictive to assume precise measurements from outside observers. 
	In this paper, we first introduce a concept of approximate pre-opacity by capturing the security level of control systems with respect to the measurement precision of the intruder. Based on this new notion of pre-opacity, we propose a verification approach for continuous-space control systems by leveraging abstraction-based techniques. In particular, a new concept of approximate pre-opacity preserving simulation relation is introduced to characterize the distance between two systems in terms of preserving pre-opacity. This new system relation allows us to verify pre-opacity of complex continuous-space control systems using their finite abstractions. We also present a method to construct pre-opacity preserving finite abstractions for a class of discrete-time control systems under certain stability assumptions.  
\end{abstract}

\begin{IEEEkeywords}
Discrete Event Systems, Opacity, Formal Abstractions
\end{IEEEkeywords}

		\section{Introduction}

Cyber-physical systems (CPS) are the technological backbone of the increasingly interconnected and smart world where security vulnerability can be catastrophic. However, the tight interaction between embedded control software and the physical environment in CPS may expose numerous attack surfaces for malicious exploitation.
In the last decade, the analysis of various security properties for CPS has drawn considerable attention in the literature \cite{basilio2021analysis,liu2022secure}.
The concept of opacity was originally introduced in computer science literature\cite{mazare2004using} for the analysis of cryptographic protocols. Afterwards, opacity was widely investigated in the domain of discrete-event systems (DES) since it allows researchers to analyze the information-flow security for dynamical systems in a formal way \cite{lafortune2018history}.
Roughly speaking, opacity is a confidentiality property that characterizes whether or not a dynamical system will reveal some potentially sensitive behavior to an external malicious observer (intruder) based on the information flow.

In the past decades, different notions of opacity were proposed in the literature to capture different security requirements in the context of DES, including language-based notions in \cite{lin2011opacity} and state-based notions in \cite{saboori2011verification,lennartson2022state,tong2022verification}. The recent results in \cite{balun2021comparing,wintenberg2022general} show that these notions are transformable to each other.
Corresponding to the different opacity notions, various verification and synthesis approaches were also developed in the DES literature; see \cite{lafortune2018history,lin2011opacity,jacob2016overview,liu2022secure,ma2021verification} and the references therein. 
Although the majority of the above-mentioned works on opacity are applied to DES models with discrete state sets, the analysis of opacity for control systems with continuous state sets has become the subject of many studies recently \cite{liu2022secure,an2019opacity,yin2020approximate}.
In particular, a new concept of \emph{approximate opacity} is proposed in \cite{yin2020approximate} which is more applicable to control systems since it allows us to quantitatively evaluate the security level of control systems whose outputs are physical signals.
More recently, a new concept of opacity, called \emph{pre-opacity}, was proposed in \cite{yang2020secure} to characterize whether or not the secret intention of the system can be revealed. In other words, different from the other opacity notions which consider the current or past secret behaviors of the system, pre-opacity captures whether or not an outside observer can be prematurely certain that the system will conduct some secret behaviors in the future. In fact, in many practical scenarios, systems are indeed more interested in hiding their intentions to do something particularly important in the future. 
Nevertheless, the results developed in  \cite{yang2020secure} are again tailored to DES models with discrete state sets, which prevents it from being applied to real-world CPS with continuous state sets.

{\bf{Our contribution.}}
In this paper, we consider the problem of verifying pre-opacity for discrete-time control systems. 
Motivated by the limitations of the results in \cite{yang2020secure}, we first introduce a new concept called \emph{approximate $K$-step pre-opacity} which is more applicable to control systems. To be more specific, unlike discrete-event systems whose state sets are discrete and outputs are logic events, control systems are in general metric systems whose state and output sets are physical signals. Therefore, the notion of pre-opacity in \cite{yang2020secure} is too restrictive by assuming that one can always precisely distinguish between two outputs in the context of control systems.
Note that 
the verification of pre-opacity for control systems is in general undecidable. 
In this work, we propose an abstraction-based pre-opacity verification approach for continuous-space control systems. 
In particular, we first propose a notion of \emph{approximate $ K $-step pre-opacity preserving simulation relation}, which is a system relation that can be used to characterize the closeness between two systems in terms of preserving pre-opacity. Based on this system relation, one can verify pre-opacity of a complex control system using its finite abstraction, instead of directly applying verification algorithms on the original control system which is undecidable. Moreover, for the class of incrementally input-to-state stable nonlinear control systems, we show that one can always construct finite abstractions which preserve pre-opacity of the control systems. 
The proposed abstraction-based methodology is the first in the literature that provides a sound way for verifying pre-opacity of discrete-time control systems with  continuous state spaces.


\section{Preliminaries}\label{sec:pre}

\subsection{Notation}
We denote by $\N$ and $\R$ the set of non-negative integers and real numbers, respectively. They are annotated with subscripts to restrict them in the usual way, e.g., $\mathbb{R}_{\geq 0}$ denotes the set of non-negative real numbers. 
Given a vector \mbox{$x\in\mathbb{R}^{n}$}, we denote by $\Vert x\Vert$ the infinity norm of $x$.
A set $B\subseteq \R^m$ is called a
{\em box} if $B = \prod_{i=1}^m [c_i, d_i]$, where $c_i,d_i\in \R$ with $c_i < d_i$ for each $i\in\{1,\ldots,m\}$.
%
For any set $A = \bigcup_{j=1}^M A_j$ of the form of finite uion of boxes, where $A_j = \prod_{i=1}^n [c_i^j, d_i^j]$, we define $\boxspan(A) = \min\{\boxspan(A_j)\mid j=1,\ldots,M\}$, where $\boxspan(A_j) = \min\{ | d_i^j - c_i^j| \mid i=1,\ldots,m\}$.
For any $\mu \leq \boxspan(A)$, define $[A]_\mu = \bigcup_{j=1}^M [A_j]_\mu$,
where
$[A_j]_\mu = [\R^m]_{\mu}\cap{A_j}$ and $[\R^m]_{\mu}=\{a\in \R^m\mid a_{i}=k_{i}\mu,k_{i}\in\mathbb{Z},i=1,\ldots,m\}$.
We denote the different classes of comparison functions by $\mathcal{K}$, $\mathcal{K}_{\infty}$ and $\mathcal{KL}$, where $\mathcal{K} \!=\! \{\gamma : \mathbb R_{\ge 0}\rightarrow\mathbb R_{\ge 0}   :  \gamma \text{ is continuous, strictly increasing and } \gamma(0)=0\}$; $\mathcal{K}_{\infty} \!=\! \{\gamma \!\in\! \mathcal{K}  \! :\! \lim_{r\rightarrow \infty}\gamma(r) \!=\!\infty\}$; $\mathcal{KL} \!=\! \{\beta : \mathbb R_{\ge 0} \!\times \mathbb R_{\ge 0} \rightarrow\mathbb R_{\ge 0}  :$ for each fixed $s$, the map  $\beta(r,s)$  belongs to class  $\mathcal{K}$  with  respect to  $r$  and, for each fixed  nonzero $r$,  the map $\beta(r,s)$ is decreasing with respect to  $s$  and $\beta(r,s) \rightarrow 0 \text{ as } s \rightarrow \infty \}$.	

\subsection{System Model}

In this paper, the system model that can be used to describe both continuous-space and finite control systems is  a \emph{tuple} 
\[S =(X, X_0, U, \rTo, Y, H ),\]
where
$X$ is a (possibly infinite) set of states,  
$X_0  \subseteq  X$ is a (possibly infinite) set of initial states,
$U$ is a (possibly infinite) set of inputs, 
$\rTo \subseteq  X \times U \times X$ is a transition relation,  
$Y$ is a (possibly infinite) set of outputs, and 
$H: X \to Y$ is the output function.
For the sake of simplicity,  we   also denote a transition $(x,u,x')\in \!\rTo\!$ by $x\rTo{u}x'$, 
where we say that $x'$ is a $u$-successor, or simply successor, of $x$. For each state $x\in X$, we denote by $U(x)$ the set of all inputs defined at $x$, i.e., $U(x)=\{u\in U: \exists x'\in X\text{ s.t. }x\rTo{u}x'\}$, and by $U^{post}_u(x)$ the set of $u$-successors of state $x$.
%
A system $S$ is said to be
\begin{itemize}
	\item \textit{metric}, if the output set $Y$ is equipped with a metric
	$\mathbf{d}:Y\times Y\rightarrow\mathbb{R}_{\geq 0}$;
	\item \textit{finite} (or \textit{symbolic}), if $X$ and $U$ are finite sets;
\end{itemize}
A finite state run of a system $S$ generated from initial state $x_0 \in X_0$ under input sequence $u_1\cdots u_n$ is a sequence of transitions
$ x_{0} \rTo{u_1} x_1\rTo{u_2} \cdots\rTo{u_n} x_n $, where $x_i\rTo{u_{i+1}}x_{i+1}$ for all $0\leq i \leq n-1$. 
The corresponding output run is a sequence of outputs $H(x_0)H(x_1)\cdots H(x_n)$.


\subsection{Exact Pre-Opacity}

In many scenarios, the system wants to hide its intention to reach some $secret$ states at some future instants in the presence of a malicious intruder (outside observer).  
In this article, we adopt a state-based formulation of secrets. Specifically, we assume that $ X_S \subseteq X$ is a set of secret states. In the sequel, we incorporate the secret state set $X_S$ in the system definition and use  $S =(X, X_{0}, X_S, U,  \!\!\rTo \! \!,   Y,  H )$ to denote a metric system. 
We consider that the intruder knows the dynamics of the system and can  observe the output sequences of the system, but cannot actively affect the behavior of the system.
To characterize whether or not the secret intention of a system can be revealed, notions of pre-opacity are proposed in  \cite{yang2020secure}.
Let us review the notion of $K$-step pre-opacity introduced in \cite{yang2020secure} as follows.	


\begin{definition}\label{EXACT}  
	Consider a system $S  = (X, X_{0}, X_S, U,$  $\!\!\rTo \! \!,   Y,  H )$ and a constant $K \in \mathbb{N}$. 
	We say that $ S $ is 
	\textit{$K$-step}  \textit{pre-opaque}
	if for any finite sequence $x_0\rTo^{u_1}x_1\rTo^{u_2} \cdots \rTo^{u_{n}}x_n$, any non-negative integer $ t \geq K $,
	there exist a finite sequence $x_0'\rTo^{u_1'}x_1'$ $\rTo^{u_2'}\cdots\rTo^{u_{n}'}x_n'\rTo^{u_{n}'}\cdots\rTo^{u_{n+t}'}x_{n+t}'$ such that  
	\begin{align*}
	H(x_i)= H(x_i'),\forall  i=\{0,\dots,n\},
	\end{align*}
	and $x_{n+t}'\notin X_S$.
\end{definition}	
Intuitively, pre-opacity requires that the intruder can never predict that the system will visit a secret state for some specific future instant. 
The above definition of $K$-step pre-opacity requires that for any behavior of the system and any $t \geq K$, there exists a behavior whose prefix generates \emph{exactly} the same output and will reach a non-secret state in exact $t$ steps. Thus, in the remainder part of the paper, we will refer to this definition as \emph{exact pre-opacity}.

\subsection{Approximate Pre-Opacity}
The notion of exact pre-opacity introduced in the previous subsection 
essentially assumes that the intruder can always measure each output or distinguish
between two different outputs precisely.  However, for metric systems whose outputs are physical signals, due to the imperfect measurement precision of potential outside observers (which is the case for almost all physical systems), it is very difficult to distinguish two observations if their difference is very small. Therefore, in the following definition, we propose a weak and ``robust" version of pre-opacity called \emph{$\delta$-approximate pre-opacity} which is more applicable to metric systems.  

\begin{definition}\label{INO}  
	Consider a system $ S =(X, X_{0}, X_S, U,$  $\!\!\rTo \! \!,   Y,  H )$ and a constant $\delta \in \mathbb{R}_{\geq 0}$. 
	We say that $ S $ is 
	\textit{$K$-step $\delta$-approximate pre-opaque}
	if for any finite sequence $x_0\rTo^{u_1}x_1\rTo^{u_2} \cdots \rTo^{u_{n}}x_n$, for any non-negative integer $ t \geq K $,
	there exist a finite sequence $x_0'\rTo^{u_1'}x_1'\rTo^{u_2'}\cdots\rTo^{u_{n}'}x_n'\rTo^{u_{n}'}\cdots\rTo^{u_{n+t}'}x_{n+t}'$
	such that 
	\begin{align*}			\underset{i\in\{0,\ldots,n\}}{\text{max}} \textbf{d}(H(x_i), H(x_i'))\leq \delta,
	\end{align*}
	and $x_{n+t}'\notin X_S$.
\end{definition}
As one can easily see, when $\delta\ =\ 0$, approximate pre-opacity boils down to the exact version in Definition~\ref{EXACT}. 
We use the following example to illustrate the notions of exact and approximate pre-opacity. 

%
\begin{example}
	Consider system $S=(X,X_0,X_S,U,\rTo,$ $Y,H)$ shown in Figure~\ref{exmaple1},
	where $X=\{A,B,C,D,E,F,G,H\}$, $X_0=\{A,E\}$, $X_S=\{C,H\}$, $U=\{u,u'\} $, 
	$Y=\{1.1,1.2,2.1,2.3,2.9,3.1,4.0,4.2\}\subseteq \R$ equipped with
	metric  $ \textbf{d} $  defined by $\textbf{ d} (y_1, y_2) = |y_1-y_2|$,  $\forall y_1, y_2 \in Y $.  We mark all secret states by red and the output of each state is specified by a value associated to it. 
	First, one can easily check that $S$ is not exact $ K $-step pre-opaque for any $K \in \N$, since we know immediately that the system is at secret state
	when value $3.1$ or $4.0$ is observed.
	Next, consider an intruder with measurement precision $\delta=0.2$. We claim that $ S $ is $0.2$-approximate $1$-step pre-opaque. For example, consider a finite path $ A\rTo^{u}B $ which  generates output path [1.1][2.3] and will reach a secret state in $1$ step. However, the intruder cannot predict for sure that the system will be at a secret state in $1$ step since there is another path $ E\rTo^{u}F $ generating a indistinguishable output path[1.2][2.1], but will reach a non-secret state $G \notin X_S$. 
	Similarly, when observing [1.2][2.1] (generated by the finite path $E\rTo^{u}F$), the intruder cannot predict for sure that the system will be at a secret state after $2$ steps either, since there exists another path $A\rTo^{u}B\rTo^{u}C\rTo^{u}D$ which will reach non-secret state $D$ in $2$ steps. This protects the possible secret intention of executing $E\rTo^{u}F\rTo^{u}G\rTo^{u}H$. 
	
	\begin{figure}
		\centering
		\includegraphics[width=0.28\textwidth]{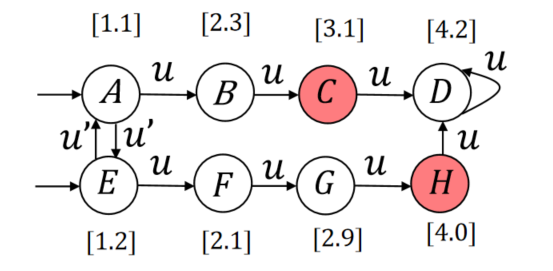}
		\caption{Example to illustrate $\delta$-approximate  $K$-step pre-opacity.} 
		\label{exmaple1}
		\vspace{-0.5cm}
	\end{figure}
\end{example}


\section{Verification of Approximate Pre-Opacity in Finite Systems}
In this section, we show how to verify approximate $K$-step pre-opacity in finite systems. Specifically, we present a necessary and sufficient conditions for $K$-step instant pre-opacity that can be checked by combining the current-state estimation together
with the reachability analysis.

In order to verify $\delta$-approximate $K$-step pre-opacity,
we need to first construct a new system called the \emph{$\delta$-approximate current-state estimator} defined  as follows.
\begin{definition}
	Let $S=(X,X_0,X_S,U,\rTo,Y,H)$ be a metric system, with the metric $\mathbf{d}$ defined over the output set, and a constant $\delta\geq 0$.
	The $\delta$-approximate current-state estimator is a system (without outputs)
	\[
	Obs(S)=(X_{obs},X_{obs,0},U,\rTo_{Obs}),
	\]
	where
	\begin{itemize}
		\item
		$X_{Obs}\subseteq X\times 2^X$ is the set of states;
		\item
		$X_{Obs,0}=\{(x,q)\!\in\! X_0\!\times\! 2^{X_0}:  x'\!\in\! q\Leftrightarrow \mathbf{d}(H(x),H(x'))\!\leq\! \delta    \}$ is the set of initial states;
		\item
		$U$ is the set of inputs, which is the same as the one in $S$;
		\item
		$\rTo_{Obs}\subseteq X_{Obs}\times U\times X_{Obs}$ is the transition function defined by:
		for any $(x,q),(x',q')\in X\times 2^X$ and $u\in U$,   $(x,q)\rTo^{u}_{Obs}(x',q')$ if
		\begin{enumerate}
			\item
			$(x,u,x')\in \rTo$; and
			\item
			$q'\!=\!\cup_{\hat{u}\in U} U^{post}_{\hat{u}}(x) \!\cap\! \{x''\!\in\! X:\!   \mathbf{d}(H(x'),H(x''))\!\leq\! \delta  \}$.
		\end{enumerate}
	\end{itemize}

Intuitively, this observer show us all the possible current state according to the output path until now. For the sake of simplicity, we only consider the part of $Obs(S)$ that is reachable from initial states.
\end{definition}

According to the definition of $ K $-step pre-opacity, we should consider the state status after $n$-step later, $ n\geq K $. Thus, we introduce a concept of $n$-step indicator $\mathcal{T}_n$ that can predicts the state status with respect to the secret sets after $n$ steps later. Specifically, one can find such $n$-step indicator by backtracking $ n $ steps from the set of all secret states.
Formally, we first define an operator $F: 2^X \rTo 2^X$ by: 
\[\forall q\in 2^X:F(q)=\{x\in X : \forall u\in U(x) ~\text{s.t.}~ x\rTo{u}x'\in q \}\]
Then, one can compute $n$-step indicator $\mathcal{T}_n$ by:
 \[\mathcal{T}_n= F^n(\mathcal{T}_0) ~\text{with}~ \mathcal{T}_0=X_s\]

We use the following result to state the main properties of $Obs(S)$.
\begin{proposition}\label{prop:obs}\cite{yin2020approximate}
	Let $S=(X,X_0,X_S,U,\rTo,Y,H)$ be a metric system, with the metric $\mathbf{d}$ defined over the output set, and a constant $\delta\geq 0$.
	Let $Obs(S)=(X_{Obs},X_{Obs,0},U,\rTo_{Obs})$ be its $\delta$-approximate current-state estimator.
	Then for any $(x_0,q_0)\in X_{Obs,0}$ and any finite run
	\[
	(x_0,q_0)\rTo^{u_1}_{Obs}  (x_1,q_1)\rTo^{u_2}_{Obs}  \cdots  \rTo^{u_n}_{Obs}(x_n,q_n),
	\]
	we have
	\begin{enumerate}
		\item
		$x_0\rTo^{u_1}x_{1}\rTo^{u_{2}}\cdots    \rTo^{u_n}x_{n}$; and
		\item
		$q_n=\{x_n'\in X: \exists x_0'\in X_{0},\exists  x_0'\rTo^{u_1'}  x_1'\rTo^{u_{2}'}  \cdots \\ \rTo^{u_n'} x_n'
		\text{ s.t. } \max_{i\in\{0,1,\dots,n\}}\mathbf{d}( H(x_i),H(x_{i}'))\leq \delta  \}$.
	\end{enumerate}
\end{proposition}

Now, we show the result of this section by providing a verification scheme for $\delta$-approximate $K$-step pre-opacity of finite metric systems.

\begin{lemma}\label{lemma:K}
	Let $S=(X,X_0,X_S,U,\rTo,Y,H)$ be a 
	metric system, with the metric $\mathbf{d}$ defined over the output set, and a constant $\delta\geq 0$.
	Let $Obs(S)=(X_{Obs},X_{Obs,0},U,\rTo_{Obs})$ be its $\delta$-approximate current-state estimator.
	Then, $S$ is $\delta$-approximate $K$-step pre-opaque if and only if
	\begin{equation}\label{eq:thmop-k}
	\forall(x,q)\in X_{Obs},\forall n\geq K: q \not\subseteq \mathcal{T}_{n-k}.
	\end{equation}
\end{lemma}
\begin{proof}

	($\Rightarrow$)
	By contraposition:
	suppose that   there exists a run 
	\[
	(x_0,q_0)\rTo^{u_1}_{Obs}   (x_1,q_1)\rTo^{u_{2}}_{Obs}  \cdots  \rTo^{u_{m}}_{Obs}(x_{m},q_{m})
	\]
	and an interger $ n\geq K $, $q_{m} \subseteq \mathcal{T}_n$.
	By Proposition~\ref{prop:obs}, we have 
	\[
	q_{m}=\left\{x_0'\in X:
	\begin{array}{c c}
	\exists  x_0'\rTo^{u_{1}'}  x_1'\rTo^{u_{2}'} \cdots  \rTo^{u_{m}'} x_{m}' \text{ s.t. }\\
	\max_{i\in\{0,1,\dots,{m}\}}\mathbf{d}( H(x_i),H(x_{m-i}'))\leq \delta
	\end{array}
	\right\}.
	\]
	Let us consider the sequence 
	\[
	(x_0',q'_0)\rTo^{u_1'}_{Obs}   (x_1',q_1')\rTo^{u_{2'}}_{Obs}  \cdots  \rTo^{u_{m}'}_{Obs}(x_{m}',q_{m}')
	\]
	Since $q_{m} \subseteq \mathcal{T}_n$ and Proposition~\ref{prop:obs}, we have $q'_{m} \subseteq \mathcal{T}_n$, i.e., $\forall  x_0'\rTo^{u_{1}'}  x_1'\rTo^{u_{2}'} \cdots  \rTo^{u_{m}'} x_{m}'\cdots  \rTo^{u_{m+n}'} x_{m+n}'$, $x_{m+n}' \in X_S$ holds. This means that the system is not $\delta$-approximate $K$-step pre-opaque.
	
%
%
%
	
	($\Leftarrow$)
	By contradiction: suppose that Equation~(\ref{eq:thmop-k}) holds and
	assume  that $S$ is not $\delta$-approximate $K$-step pre-opaque.
	Then,   there exists  an initial state $x_0\in X_0$
	and a sequence of transitions
	$x_0\rTo^{u_1}x_1\rTo^{u_2}\cdots\rTo^{u_m}x_m$
	such that
	there exist
	an initial state $x_0'\in X_0$, a sequence of transitions $x_0'\rTo^{u_1'}x_1'\rTo^{u_2'}\cdots\rTo^{u_m'}x_m'\rTo^{u_{m+1}'}\cdots\rTo^{u_{m+n}'}x_{m+n}'\in X_s$ and an integer $n \geq K$
	such that
	$\max_{i\in\{0,1,\dots,n\}}\mathbf{d}( H(x_i),H(x_i'))\leq \delta$.
	Let us consider the following sequence of transitions in $S_{Obs}$
	\[
	(x_0',q_0)\rTo^{u_1'}_{Obs}   \cdots  \rTo^{u_m'}_{Obs}(x_m',q_m)\rTo^{u_{m+1}'}_{Obs}\cdots\rTo^{u_{m+n}'}_{Obs}(x_{m+n}',q_{m+n}).
	\]
	Then we have $ q_{m+n}\subseteq X_S $, i.e., $q_{m}\subseteq \mathcal{T}_n$. 
	This violate the Equation~(\ref{eq:thmop-k}) holds, i.e.,  $S$ has to be $\delta$-approximate $ K $-step pre-opaque.
\end{proof}
	
	Lemma~\ref{lemma:K} seems provide a way to verify $K$-step pre-opacity. However, it still cannot be directly used for the
verification of $K$-step pre-opacity at finite systems. The main issue is that we need to check whether or not the $ q\nsubseteq \mathcal{T}_n$  for any $n \geq K$, which has infinite number of instants. In the following result, we show this issue can be solved.
\begin{theorem}\label{thm:n}
	Let $S=(X,X_0,X_S,U,\rTo,Y,H)$ be a 
	metric system, with the metric $\mathbf{d}$ defined over the output set, and a constant $\delta\geq 0$.
	Let $Obs(S)=(X_{Obs},X_{Obs,0},U,\rTo_{Obs})$ be its $\delta$-approximate current-state estimator.
	Then, $S$ is $\delta$-approximate $k$-step pre-opaque if and only if
	\begin{equation}\label{eq:thmop-n}
	\forall(x,q)\in X_{Obs}: q \not\subseteq \mathcal{T}_k.
	\end{equation}
\end{theorem}
\begin{proof}
	
	($\Rightarrow$)
	
	The necessity follows directly from Lemma~\ref{lemma:K}.

	($\Leftarrow$)
	By contradiction: suppose that Equation~(\ref{eq:thmop-n}) holds and
	assume  that $S$ is not $\delta$-approximate $K$-step pre-opaque.
	Then,   by Lemma~\ref{lemma:K}, we know there exists  an initial state $x_0\in X_0$
	and a sequence of transitions
	$x_0\rTo^{u_1}x_1\rTo^{u_2}\cdots\rTo^{u_m}x_m$
	such that
	there exist
	an initial state $x_0'\in X_0$, a sequence of transitions $x_0'\rTo^{u_1'}x_1'\rTo^{u_2'}\cdots\rTo^{u_m'}x_m'\rTo^{u_{m+1}'}\cdots\rTo^{u_{m+n}'}x_{m+n}'\in X_s$ and an integer $n \geq K$
	such that
	$\max_{i\in\{0,1,\dots,n\}}\mathbf{d}( H(x_i),H(x_i'))\leq \delta$.
	Let us consider the following sequence of transitions in $S_{Obs}$
	\[
	(x_0',q_0)\rTo^{u_1'}_{Obs}   \cdots  \rTo^{u_m'}_{Obs}(x_m',q_m)\rTo^{u_{m+1}'}_{Obs}\cdots\rTo^{u_{m+n}'}_{Obs}(x_{m+n}',q_{m+n}).
	\]
	Then we have $ q_{m+n}\subseteq X_S $, i.e., $q_{m}\subseteq \mathcal{T}_n$. 
	Since $n\geq K$ in Lemma~\ref{lemma:K}, we can find  the prefix of it:
	\[
	(x_0',q_0)\rTo^{u_1'}_{Obs}   \cdots  (x_m',q_m)\rTo^{u_{m+1}'}_{Obs}\cdots\rTo^{u_{m+n-k}'}_{Obs}(x_{m+n-k}',q_{m+n-k}).
	\]. Then we know 	$x_0'\rTo^{u_1'}x_1'\rTo^{u_2'}\cdots\rTo^{u_{m+n-K}'}x_{m+n-k}'\in X_s$, i.e., $q_{m+n-k} \subseteq \mathcal{T}_k$.
	This violate the Equation~(\ref{eq:thmop-n}) holds, i.e.,  $S$ has to be $\delta$-approximate $ K $-step pre-opaque.

\end{proof}

\section{Approximate Simulation Relation for $K$-Step  Pre-Opacity}\label{sec:3}

In the last section, we introduced notions of exact and approximate pre-opacity for control systems. However, the (approximate) pre-opacity is in general hard (or even infeasible) to check for control systems since there is no systematic way in the literature to check pre-opacity for systems with infinite state sets so far. On the other hand, existing tools and algorithms (such as \cite{yang2020secure}) in DES literature can be leveraged to check pre-opacity for finite systems. 
Therefore, to solve the pre-opacity verification problem for control systems, 
it would be more feasible to verify pre-opacity on their finite abstractions and then carry back the result to the concrete ones. The key to the construction of such finite abstraction is the establishment of formal relations between the concrete and abstract systems. 

In this section, we first propose a new system relation called \emph{approximate $K$-step pre-opacity preserving simulation relation}, and then show the usefulness of the proposed system relation in terms of verifying pre-opacity.

\begin{definition}(Approximate $K$-step Pre-Opacity Preserving Simulation Relation)\label{AKP}
	Consider two metric systems $S_{a}=(X_{a},X_{a0},X_{aS},U_{a},\rTo_{a},Y_a,H_{a})$ and $S_{b}=(X_{b},X_{b0},X_{bS},U_{b},\rTo_{b},Y_b,H_{b})$ with the
	same output sets $Y_a=Y_b$ and metric $\mathbf{d}$.
	Given $\varepsilon\in\mathbb{R}_{\geq 0}$, a relation \mbox{$R\subseteq X_{a}\times X_{b}$} is called an $\varepsilon$-approximate $K$-step pre-opacity preserving simulation relation ($\varepsilon$-AKP simulation relation) from $S_{a}$ to $S_{b}$ 
	if
	\begin{enumerate}
		\item
		\begin{enumerate}
			\item
			$\forall x_{a0}\in X_{a0},\exists x_{b0}\in X_{b0}: (x_{a0},x_{b0})\in R$;
			\item
			$\forall x_{b0}\in X_{b0},\exists x_{a0}\in X_{a0}: (x_{a0},x_{b0})\in R$;
		\end{enumerate}
		\item
		$\forall (x_a,x_b)\in R:\mathbf{d}(H_{a}(x_{a}),H_{b}(x_{b}))\leq\varepsilon$;	
		\item
		For any $(x_a,x_b)\in R$, we have
		\begin{enumerate}
			\item
			$\forall x_a\rTo_{a}^{u_a} x_a',\exists x_b\rTo_{b}^{u_b} x_b':(x_a',x_b')\in R$;
			\item
			$\forall x_b\rTo_{b}^{u_b} x_b',\exists x_a\rTo_{a}^{u_a} x_a':(x_a',x_b')\in R$.
			\item
			$\forall x_b\rTo_{b}^{u_b} x_b'\in X_b\setminus X_{bS},\exists x_a\rTo_{a}^{u_a} x_a'\in X_a\setminus X_{aS}:(x_a',x_b')\in R$.
		\end{enumerate}
	\end{enumerate}
	We say that $S_{a}$ is $\varepsilon$-AKP simulated by $S_{b}$, denoted by $S_a\preceq_A^\varepsilon S_b$, 
	if there exists an $\varepsilon$-AKP  simulation relation $R$ from $S_a$ to $S_b$. A (finite) system $S_{b}$ that simulates $S_{a}$ through the $\varepsilon$-AKP simulation relation is called a pre-opacity preserving (finite) abstraction of $S_{a}$.	Note that the proposed $\varepsilon$-AKP  simulation relation is still a one-sided relation because conditions 1) and 3) are asymmetric.
\end{definition}


The following theorem shows how to use the above proposed simulation relation in terms of verifying pre-opacity. 	
\begin{theorem}\label{thm:ASOP}
	Consider two metric systems  $S_{a}=(X_{a},X_{a0},$ $X_{aS},U_{a},\rTo_{a},Y_a,H_{a})$ and $S_{b}=(X_{b},X_{b0},X_{bS},U_{b},$ $\rTo_{b},Y_b,H_{b})$ with the same output sets $Y_a=Y_b$ and metric $\mathbf{d}$ and
	let $\varepsilon,\delta\in\mathbb{R}_{\geq 0}$.
	If  $S_a\preceq_A^\varepsilon S_b$, 
	then we have:
	\begin{align}
	&S_b\text{ is $\delta$-approximate $ K $-step pre-opaque} \nonumber\\
	\Rightarrow  &S_a \text{ is ($\delta+2\varepsilon$)-approximate $ K $-step pre-opaque}\nonumber.
	\end{align}
\end{theorem}
\begin{proof}
	Let us consider an arbitrary initial state $x_0\in X_{a0}$, an arbitrary finite run
	$x_0\rTo_a^{u_1} x_1\rTo_a^{u_2}\cdots\rTo_a^{u_n}x_n$ in $S_a$, and any non-negative integer $ t \geq K $.
	Since  $S_a\preceq_A^\varepsilon S_b$, by conditions 1)-a), 2) and 3)-a)  in Definition~\ref{AKP}, there exists an initial state $x_0'\in X_{b0}$ and  a finite run
	$x_0'\rTo_b^{u_1'}x_1'\rTo_b^{u_2'}\cdots\rTo_b^{u_n'}x_n'$ in $S_b$ such that 
	\begin{equation}\label{eq:ep1c}
	\forall i\in \{0,\dots, n\}: \mathbf{d}( H_a(x_i),H_b(x_i')) \leq \varepsilon.
	\end{equation}
	Since $S_b$ is $\delta$-approximate $K$-step  pre-opaque, by Definition~\ref{INO}, for any non-negative integer $ t \geq K $, there exist an initial state $x_0''\in X_{b0}$ and
	a finite run $x_0''\rTo_b^{u_1''}x_1''\rTo_b^{u_2''}\cdots\rTo_b^{u_n''}x_n''\rTo_b^{u_{n}''}x_{n+1}''\cdots\rTo_b^{u_{n+t}''}x_{n+t}''$ such that  $x_{n+t}''\in X_b\setminus X_{bS}$ and
	\begin{equation}\label{eq:ep2c}
	\max_{i\in\{0,\dots,n\}}\mathbf{d}( H_b(x_i'),H_b(x_i''))\leq \delta.
	\end{equation}
	Again, since $S_a\preceq_A^\varepsilon S_b$, by conditions 1)-b), 2), 3)-b) and 3)-c) in Definition~\ref{AKP}, there exists an  initial state  $x_0'''\in X_{a0}$ 	
	and a finite run  $x_0'''\rTo_a^{u_1'''}x_1'''\rTo_a^{u_2'''}\cdots\rTo_a^{u_n'''}x_n'''\rTo_a^{u_{n}'''}x_{n+1}'''\cdots\rTo_a^{u_{n+t}'''}x_{n+t}'''$
	such that  $x_{n+t}'''\in X_a\setminus X_{aS}$  and
	\begin{equation}\label{eq:ep3c}
	\forall i\in \{0,\dots, n+t\}:\mathbf{d}( H_a(x_i'''),H_b(x_i'')) \leq \varepsilon.
	\end{equation}
	Combining inequalities~(\ref{eq:ep1c}),~(\ref{eq:ep2c}),~(\ref{eq:ep3c}), and using the triangle inequality,
	we have
	\begin{equation}\label{eq:ep4c}
	\max_{i\in\{0,\dots,n\}}\mathbf{d}( H_a(x_i),H_a(x_i''')) \leq \delta+2\varepsilon.
	\end{equation}
	Since $x_0\in X_{a0}$ and $x_0\rTo_a^{u_1} x_1\rTo_a^{u_2}\cdots\rTo_a^{u_n}x_n$ are arbitrary, we conclude that
	$S_a$ is $(\delta+2\varepsilon)$-approximate $ K $-step pre-opaque.
\end{proof}

Essentially, Theorem~\ref{thm:ASOP} provides us with a sufficient condition for verifying pre-opacity of control systems using abstraction-based techniques. In particular, when encountered with a complex control system $ S_a $ (possibly with infinite state set), one can build a finite abstraction $ S_b $ for $ S_a $ through the proposed $\varepsilon$-AKP simulation relation. 
Then, one can verify pre-opacity of the  finite abstraction $ S_b $ leveraging existing algorithms in DES literature, and then carry back the verification result to the concrete system $ S_a $ by employing the result obtained in  Theorem~\ref{thm:ASOP}. 
Note that such $ \delta $ and $ \varepsilon $ are parameters that specify two different types of precision. The parameter $ \delta $  is used to specify the intruder's measurement precision under which one can guarantee pre-opacity of a single system, whereas $ \varepsilon $ appeared in the proposed $\varepsilon$-AKP simulation relation is used to describe the ``distance'' between two systems in terms of preserving pre-opacity. Besides, The reader should notice that $ \delta $ relaxation is a single-sided condition only from $S_b$ to $S_a$.

We illustrate the newly proposed $\varepsilon$-AKP simulation relation and the preservation of pre-opacity between two related finite systems by the following example.

\begin{example}
	Consider systems $S_a$ and $S_b$ shown in Figures~\ref{fig:4} and~\ref{fig:5}, respectively.
	All secret states are marked by red and the output of each state is specified by the value associated to it.
	Let us consider the following relation
	$R=\{(A,L),$ $(B,I),(C,I),(D,I),(E,J),(F,J),(G,J),(H,K)\}$.
	We claim that $R$ is an $\varepsilon$-approximate $ K $-step pre-opacity preserving simulation relation from $S_{a}$ to $S_{b}$ when $\varepsilon=0.1$. 
	First, for both initial states $A$ and $H$ in $S_{a}$, we have $L, K \in X_{b0}$ in $S_{b}$ such that $(A, L) \in R$ and $(H, K)\in R$. Thus,
	condition~1) in Definition~\ref{AKP} holds.
	Also, one can easily check that $\mathbf{d}(H_a(x_a),H_b(x_b))\leq 0.1$ for any $(x_a,x_b)\in R$.
	Therefore, condition~2) in Definition~\ref{AKP} holds.
	Moreover, one can easily check that  condition~3)-a), 3)-b) in Definition~\ref{AKP} holds as well.
	For example, 
	for  $(B,I)\in R$ and $B\!\rTo_{a}^{u}\!C$, we can choose $I\!\rTo_{b}^{u}\!I$ such that $(C,I)\in R$.
	Finally, condition~3)-c) in Definition~\ref{AKP} is also satisfied.  As an example, for $(H,K) \in R$ and the transition $K\!\rTo_{b}^{u}\!J \in X_b \setminus X_{bS}$ in $S_b$, there exists a transition  $H\!\rTo_{b}^{u}\!G \in X_a \setminus X_{aS}$ in $S_a$ such that $(G, J) \in R$.
	Therefore, one can conclude that $R$ is an $\varepsilon$-AKP simulation relation from $S_{a}$ to $S_{b}$, i.e.,
	$S_a\preceq_A^{0.1} S_b$. 
	Furthermore, it can be easily seen that $S_b$ is $\delta$-approximate 0-step pre-opaque with $\delta=0.2$. 
	Therefore, according to Theorem~\ref{thm:ASOP}, we can readily conclude that $S_a$ is $0.4$-approximate 0-step pre-opaque, where $0.4=\delta+2\varepsilon$, without applying any verification algorithm to $S_a$ directly.
	\begin{figure}
		\centering
		\subfigure[$S_{a}$]{\label{fig:4}
			\includegraphics[width=0.22\textwidth]{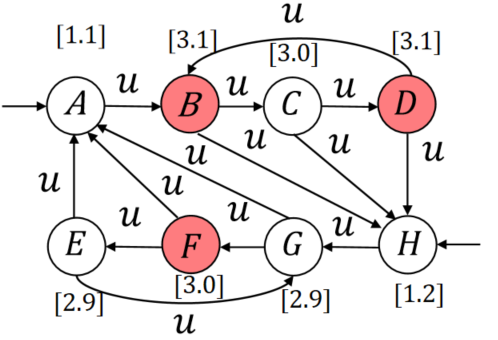}}
		\subfigure[$S_{b}$]{\label{fig:5}
			\includegraphics[width=0.14\textwidth]{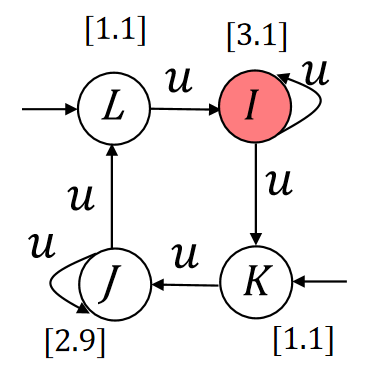}}
		\caption{Example of $\varepsilon$-approximate 0-step pre-opacity preserving simulation relation.} \label{exmaple2}
		\vspace{-0.5cm}
	\end{figure}
\end{example}


\section{Pre-Opacity of Control Systems}\label{sec:cs}
In the previous section, we introduced the concept of approximate pre-opacity preserving simulation relations and discussed how it can be used to solve pre-opacity verification problem for (possibly infinite) systems by leveraging abstraction-based techniques. 
In this section, we proceed to investigate how to construct such pre-opacity preserving finite abstractions for control systems. In particular, we show that for a class of discrete-time control systems under certain stability assumptions, one can build finite abstractions which preserve pre-opacity of the concrete control systems under the proposed $\varepsilon$-AKP simulation relation.

\subsection{Discrete-Time Control Systems}\label{subsec:dtsc}
In this section,  we consider a class of discrete-time control systems of the following form.
\begin{definition}\label{def:sys1}
	A discrete-time control system $\Sigma$ is defined by the tuple
	$\Sigma=(\mathbb X,\mathbb S,\mathbb U,f,\mathbb Y,h)$,
	where $\mathbb X$, $\mathbb S\subseteq \mathbb X$, $\mathbb U$, and $\mathbb Y$ are the state, secret state, input, and output sets, respectively.
	The map $f: \mathbb X\times \mathbb U \rightarrow \mathbb X $ is the state transition function, and $h:\mathbb X \rightarrow \mathbb Y$ is the output map.
	The dynamics of $\Sigma $ is described by difference equations of the form
	\begin{align}\label{eq:2}
	\Sigma:\left\{
	\begin{array}{rl}
	\xi(k+1)\!\!\!\!\!\!&=f(\xi(k),\upsilon(k)),\\
	\zeta(k)\!\!\!\!\!\!&=h(\xi(k)),
	\end{array}
	\right.
	\end{align}
	where $\xi:\mathbb{N}\rightarrow \mathbb X $, $\zeta:\mathbb{N}\rightarrow \mathbb Y$, and $\upsilon:\mathbb{N}\rightarrow \mathbb U$ represent the state, output, and input signals, respectively.
	We write $\xi_{x\upsilon}(k)$ to denote the point reached at time $k$ under the input signal $\upsilon$ from initial condition $x=\xi_{x\upsilon}(0)$, and $\zeta_{x\upsilon}(k)$ to denote  the output corresponding to state $\xi_{x\upsilon}(k)$, i.e. $\zeta_{x\upsilon}(k)=h(\xi_{x\upsilon}(k))$. Throughout this section, we assume that the output map satisfies the following Lipschitz condition: $\Vert h(x)-h(x')\Vert\leq\alpha(\Vert x-x'\Vert)$ for some $\alpha\in\mathcal{K}_\infty$, for all $x,x'\in\mathbb{X}$.
\end{definition}

\subsection{Construction of Finite Abstractions}

Next, we present how to construct finite abstractions which preserve pre-opacity for a class of discrete-time control systems. Specifically, the finite abstraction is built under the assumption that the concrete control system is incrementally input-to-state stable \cite{tran2018advances} as defined next. 

\begin{definition}\label{ISS}
	System $\Sigma=(\mathbb X,\mathbb S,\mathbb U,f,\mathbb Y,h)$ is called incrementally input-to-state stable ($\delta$-ISS) if there exist functions $\beta \in \mathcal{KL}$ and $\gamma \in \mathcal{K}_\infty$ such that $\forall x,x'\in \mathbb X$ and $\forall \upsilon,\upsilon' \in \mathbb{N} \rightarrow \mathbb U$, the following inequality holds for any $k\in\N$:
	\begin{align}\label{ISS_enq}
	\Vert \xi_{x\upsilon}(k)\!-\!\xi_{x'\upsilon'}(k)\Vert\!\leq\!\beta(\Vert x-x'\Vert,k)\!+\!\gamma(\Vert \upsilon-\upsilon'\Vert_\infty).
	\end{align}
\end{definition}

Next, in order to construct pre-opacity preserving finite abstractions for a control system  $\Sigma=(\mathbb X,\mathbb S,\mathbb U,f,\mathbb Y,h)$ in Definition \ref{def:sys1}, we define an associated metric system
$S(\Sigma)=(X,X_0,X_S,U,\rTo,Y,H)$,
where $X\!=\!\mathbb X$, $X_0\!=\!\mathbb X$, $X_S\!=\!\mathbb S$, $U\!=\!\mathbb U$, $Y\footnote{The output set is assumed to be equipped with the infinity norm: $\mathbf{d}(y_1,y_2)=\Vert y_1-y_2\Vert$, $\forall y_1,y_2\in Y$.} = \mathbb Y$, $H=h$, and $x\rTo^{u}x'$ if and only if $x'=f(x,u)$. 
In the sequel, we will use $ S(\Sigma) $ to denote the concrete control systems interchangeably.

Now, we are ready to introduce a finite abstraction for a control system $\Sigma=(\mathbb X,\mathbb S,\mathbb U,f,\mathbb Y,h)$. To do so, from now on, we assume that sets $\mathbb X$, $\mathbb S$ and $\mathbb U$ are of the form of finite union of boxes. Consider a tuple $\mathsf{q}=(\eta,\mu,\theta)$ of parameters, where $0<\eta\leq\min\left\{\boxspan(\mathbb{S}),\boxspan(\mathbb{X}\setminus\mathbb{S})\right\}$ is the state set quantization,  $0<\mu\leq\boxspan(\mathbb{U})$ is the input set quantization, and $ \theta $ is the designed inflation parameter. A finite abstraction of $\Sigma$ is defined as 
\begin{equation} \label{symbolicmodel}
S_\mathsf{q}(\Sigma)=(X_{\mathsf{q}},X_{\mathsf{q}0},X_{\mathsf{q}S},U_\mathsf{q},\rTo_{\mathsf{q}},Y_\mathsf{q},H_\mathsf{q}),
\end{equation}
where $X_{\mathsf{q}}=X_{\mathsf{q}0}=\left[\mathbb X\right]_\eta$, $X_{\mathsf{q}S}=\left[\mathbb S^
\theta\right]_\eta$, where $\mathbb S^{\theta} = \{x \in \mathbb X : \exists x' \in \mathbb S, \text{ s.t. } \Vert x - x'\Vert \leq \theta\}$ denotes the $\theta$-expansion of set $\mathbb{S}$,
$U_\mathsf{q}=\left[\mathbb U\right]_\mu$, $Y_\mathsf{q}=\{h(x_{\mathsf{q}})\,\,|\,\,x_{\mathsf{q}}\in X_{\mathsf{q}}\}$, $H_\mathsf{q}(x_\mathsf{q})=h(x_\mathsf{q})$, $\forall x_\mathsf{q}\in X_{\mathsf{q}}$, and
\begin{align}
\label{constructeq}
x_\mathsf{q}\rTo^{u_\mathsf{q}}_{\mathsf{q}}x'_\mathsf{q} \ \text{if and only if}\ \Vert x'_\mathsf{q}-f(x_\mathsf{q},u_\mathsf{q})\Vert\leq \eta.
\end{align}
Now, we are ready to present the main result of this section, which shows that under some condition over the quantization parameters $\eta$, 
$ \theta $ and $\mu$, the finite abstraction $S_\mathsf{q}(\Sigma)$ constructed in \eqref{symbolicmodel} indeed simulates our concrete control system $S(\Sigma)$ through approximate $ K $-step pre-opacity preserving simulation relation as in Definition~\ref{AKP}. 
\begin{theorem}\label{theorem2}
	Consider a $\delta$-ISS control system $\Sigma=(\mathbb X,\mathbb S,\mathbb U,f,\mathbb Y,h)$ as in Definition~\ref{ISS} and its associated metric system $S(\Sigma)$. For any desired precision $\varepsilon>0$, and any tuple $\mathsf{q}=(\eta,\mu,\theta)$ of quantization parameters satisfying
	\begin{align}\label{teq1}
	\beta\left(\alpha^{-1}(\varepsilon),1\right)+\gamma(\mu)+\eta&\leq\alpha^{-1}(\varepsilon),\\
	\label{teq2}
	\beta\left(\alpha^{-1}(\varepsilon),1\right)+\eta& \leq\theta,
	\end{align}
	we have $S(\Sigma) \preceq^{\varepsilon}_{A} S_\params(\Sigma) $.
\end{theorem}

\begin{proof}
	Given a desired precision $\varepsilon>0$ appeared in  Definition~\ref{AKP}, let us consider a relation $R\subseteq  X\times X_{\params}$ defined by:
	{$(x,x_{\params})\in R$}
	if and only if
	{$\Vert x-x_{\params}\Vert\leq\alpha^{-1}(\varepsilon)$}.
	First, according to the construction of $S_{\params}(\Sigma)$ in \eqref{symbolicmodel}, for any initial state $x_{0} \in X_{0}$ in $S(\Sigma)$, there exists an initial state $x_{q0} \in X_{q0}$ in $S_\params(\Sigma)$ such that $\Vert x_{a0} - x_{q0}\Vert \leq \eta$. By \eqref{teq1}, we further have $\eta \leq \alpha^{-1}(\varepsilon)$. Thus, we get that $(x_{0},x_{q0}) \in R$ and condition 1)-a) in Definition~\ref{AKP} readily holds.  
	Moreover,  for any $ x_{\params0}\!\in\! X_{\params0}$, there exists $x_0= x_{\params0}\!\in\! 	X_0$ such that $\Vert x_0-x_{\params0}\Vert=0\leq\alpha^{-1}(\varepsilon)$.
	Hence, \mbox{$(x_0,x_{\params0})\in{R}$} and condition 1)-b) in Definition \ref{AKP} is also satisfied.
	Now consider any \mbox{$(x,x_{\params})\in R$}. 
	By the definition of $R$ and the Lipschitz assumption, we have 
	$\Vert H(x)-H_{\params}(x_{\params})\Vert=\Vert h(x)-h(x_{\params})\Vert\leq\alpha(\Vert x-x_{\params}\Vert)\leq\varepsilon$,
	which shows that condition 2) in Definition~\ref{AKP} is satisfied. 	
	Further, let us proceed to prove condition 3) in Definition \ref{AKP}. First, consider any pair \mbox{$(x,x_{\params})\in R$}. 
	Given any input $u\in U$ and the transition \mbox{$x\rTo^{u} x'=f(x,u)$} in $S(\Sigma)$, let us  choose an input $ u_q\in U_q $ such that $\Vert u-u_{\params}\Vert\leq\mu$,
	where $\mu\leq\boxspan(\mathbb{U})$. 
	From the \mbox{$\delta$-ISS} assumption on $\Sigma$, the distance between $x'$ and \mbox{$f(x_{\params},u_{\params})$} is bounded as:
	\begin{align}
	\Vert  x'-f(x_{\params},u_{\params})\Vert 
	& \stackrel{\eqref{ISS_enq}}\leq\beta\left(\Vert x-x_{\params}\Vert,1\right)+\gamma\left(\Vert u-u_{\params}\Vert\right) \notag\\ 
	&\leq\beta\left(\alpha^{-1}(\varepsilon),1\right)+\gamma\left(\mu\right)\label{b02}.
	\end{align}
	Besides, by the structure of $S_\params(\Sigma)$ as in \eqref{constructeq}, we have
	\begin{equation}
	\Vert f(x_{\params},u_{\params})-x'_{\params}\Vert\leq\eta. \label{ieq5}
	\end{equation}
	Now, combining the inequalities (\ref{teq1}), (\ref{b02}), (\ref{ieq5}), and triangle inequality, we obtain:
	\begin{align*}
	\Vert x'-x'_{\params}\Vert&=\Vert x'-f(x_{\params},u_{\params})+f(x_{\params},u_{\params})-x'_{\params}\Vert\\ \notag
	&\leq\Vert x'-f(x_{\params},u_{\params})\Vert+\Vert f(x_{\params},u_{\params})-x'_{\params}\Vert\\\notag&\leq\beta\left(\alpha^{-1}(\varepsilon),1\right)+\gamma\left(\mu\right)+\eta\leq\alpha^{-1}(\varepsilon).
	\end{align*}
	Therefore, we can conclude that  \mbox{$(x',x'_{\params})\in{R}$} and condition 3)-a) in Definition \ref{AKP} holds. 
	Next, let us show that the condition 3)-b) in Definition \ref{AKP} holds as well. Consider $ x_q $ and any input $ u_q  \in U_{q} $ in $S_\params(\Sigma)$. 
	Let us choose $ u=u_q $. Then, we get the unique transition $x\rTo^{u} x'=f(x,u)$ in $S(\Sigma)$.
	Be leveraging the \mbox{$\delta$-ISS} assumption on $\Sigma$, we have that the distance between $x'$ and \mbox{$f(x_{\params},u_{\params})$} is bounded as:
	\begin{align}
	\Vert x'-f(x_{\params},u_{\params})\Vert\leq&\beta\left(\Vert x-x_{\params}\Vert,1\right)+\gamma\left(\Vert u-u_{\params}\Vert\right) \notag\\
	\leq&\beta\left(\alpha^{-1}(\varepsilon),1\right). \label{iec1}
	\end{align} 
	Based on the structure of $S_\params(\Sigma)$, there exists \mbox{$x'_{\params}\in{X}_{\params}$} s.t.:
	\begin{equation}
	\Vert f(x_{\params},u_{\params})-x'_{\params}\Vert\leq\eta, \label{iec2}
	\end{equation}
	which, by the definition of $S_\params(\Sigma)$ in \eqref{constructeq}, implies the existence of $x_{\params}\rTo^{u_{\params}}_{\params}x'_{\params}$ in $S_{\params}(\Sigma)$.
	Combining inequalities (\ref{teq1}), (\ref{iec1}), (\ref{iec2}), and triangle inequality, we obtain:
	\begin{align*}
	\Vert x'-x'_{\params}\Vert&=\Vert x'-f(x_{\params},u_{\params})+f(x_{\params},u_{\params})-x'_{\params}\Vert\\ \notag
	&\leq\Vert x'-f(x_{\params},u_{\params})\Vert+\Vert f(x_{\params},u_{\params})-x'_{\params}\Vert\\\notag&\leq\beta\left(\alpha^{-1}(\varepsilon),1\right)+\eta\leq\alpha^{-1}(\varepsilon).
	\end{align*}
	Therefore, we conclude that \mbox{$(x'_{\params},x')\in{R}$} and condition 3)-b) in Definition \ref{AKP} holds.
	Finally, let us show that condition 3)-c) in Definition \ref{AKP} holds.
	To this end, we firstly consider an arbitrary transition $x_q\rTo^{u_q} x_q'$ with $x_q'\notin X_S$ in $S_q(\Sigma)$. Similar to the proof of condition 3)-b), we can show the existence of a transition $x\rTo^{u}x'$ in $S(\Sigma)$ where $(x', x_\params ') \in R$ holds, and the input is chosen as $u=u_{\params}\in U_{\params}$. Then by the construction of the secret set in the finite abstraction, one has $X_{\params S} = [\mathbb S^{\theta}]_{\eta}$ with the inflation parameter satisfying $\theta \!\geq\! \beta\left(\alpha^{-1}(\varepsilon),1\right)\!+\!\eta$ and $0 \!<\! \eta \!\leq\! \text{min} \{span(\mathbb S),span(\mathbb{X} \setminus \mathbb S)\}$, which also implies that the size of the non-secret region in  $S_q(\Sigma)$ is smaller than that in $S(\Sigma)$. 
	Therefore, since $(x', x_\params ') \in R$ which implies $\Vert x' - x_\params ' \Vert \leq\beta\left(\alpha^{-1}(\varepsilon),1\right)+\eta \leq \theta$, we obtain that $x ' \notin X_{S}$. Thus, we conclude that condition 3)-c) in Definition \ref{AKP} holds, which completes the proof.
\end{proof}

\section{Example}\label{sec:example}
\begin{figure}
	\centering
	\includegraphics[width=0.22\textwidth]{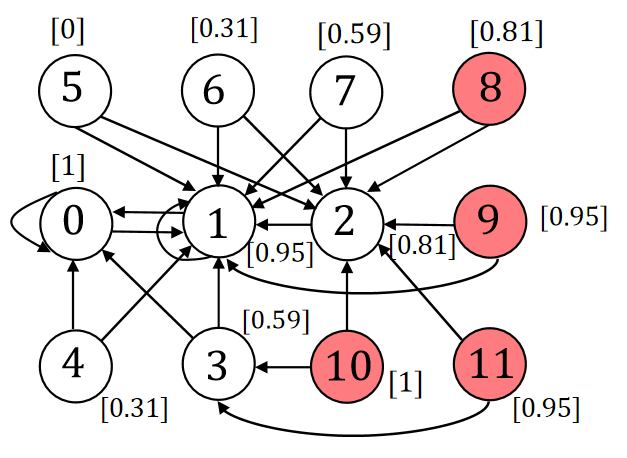}
	\caption{A $0$-step pre-opacity preserving finite abstraction of a control system.} \label{fig:casestudy}
	\vspace{-0.5cm}
\end{figure}
Here, we provide an example to illustrate the proposed abstraction-based pre-opacity verification approach. Consider the following simple control system:
\begin{align}
\Sigma:\left\{\begin{array}{l}
\xi(k+1)=0.2 \xi(k)+v(k) \\
\zeta(k)=\Vert \cos ({0.1\pi} \xi(k))\Vert,
\end{array}\right.   
\end{align}
where the state set is $\mathbb{X} = \mathbb{X}_0 = [0,12)$, the secret set is $\mathbb{X}_S = [11, 12)$, the input set is a singleton $\mathbb{U} = \{0.05\}$, and the output set is $\mathbb{Y} = [0,1]$. The output function of the system satisfies the Lipschitz condition as in Definition~\ref{def:sys1} with $\alpha(r) =0.1\pi r$. 
The main goal of the example is to verify approximate pre-opacity of the system by leveraging the proposed  abstraction-based approach.
Next, we apply our main results to achieve this goal.

First, let us construct a finite abstraction $S_\mathsf{q}(\Sigma)$ of $\Sigma$ which preserves pre-opacity with desired precision $\varepsilon=0.4$ as in Definition~\ref{AKP}. 
Note that by Definition~\ref{ISS}, one can readily check that this control system $\Sigma$ is $\delta$-ISS with $\beta(r,k) = 0.2^kr$ and $\gamma(r)= 2r$.
Next, a tuple of quantization parameters $q = (\eta,\mu,\theta) = (1,0,2.3)$ are chosen such that inequalities \eqref{teq1}-\eqref{teq2} are satisfied.  By Theorem~\ref{theorem2}, we have $S(\Sigma) \preceq^{0.4}_{A} S_\params(\Sigma)$.
Given the quantization parameters $q = (\eta,\mu,\theta) = (1,0,2.3)$, the state set $\mathbb{X}$ is discretized into $12$ discrete states as $X_q = X_{q0} = \{0,1,2,\cdots,11\}$, the discrete secret set is $X_{qS} = \{8,9,10,11\}$, the discrete input set is $U_q = \{0.05\}$, and the discrete output set is $Y_q =\{0, 0.31, 0.59, 0.81, 0.95,1\}$. The obtained finite abstraction $S_\mathsf{q}(\Sigma)$ of $\Sigma$ is shown in Fig.~\ref{fig:casestudy}. The states marked in red represent the secret states, and the output of each state is specified by a value associated to it. Note that the system can be initiated from any state since $X_q = X_{q0}$ and the input $u = 0.05$ is omitted in the figure for the sake of better presentation.
One can readily check that $S_\mathsf{q}(\Sigma)$ is exact $0$-step pre-opaque since for any run generated from any initial state of the system and any future instant $k \geq 0$, there exists another
run with exactly the same output trajectory such that it will reach a non-secret state in exactly $k$ steps. 
As an example, consider a state run $11\rTo^{u}2\rTo^{u}1\rTo^{u}0\rTo^{u}1$ which generates an output run $[0.95][0.81][0.95][1][0.95]$. There exists another state run $9\rTo^{u}2\rTo^{u}1\rTo^{u}0\rTo^{u}1$ which generates exactly the same output behavior, and will reach non-secret states (either $0$ or $1$) in any future time step $k \geq 0$.
Finally, by leveraging Theorem~\ref{thm:ASOP}, we can readily conclude that the concrete system $\Sigma$ is $0.8$-approximate $0$-step pre-opaque without directly applying verification algorithms on it.

\section{Conclusion}\label{sec:conclusion}
In this work, we proposed an abstraction-based verification framework tailored to a security property called pre-opacity for discrete-time control systems. 
The concept of pre-opacity was first extended to an approximate version which is more applicable to control systems with continuous-space outputs. Then, a notion of approximate pre-opacity preserving simulation relation was proposed, based on which one can verify pre-opacity of control systems using their finite abstractions. We also investigated how to construct finite abstractions that preserves pre-opacity for a class of control systems via the proposed system relation.  Finally, an example was presented to illustrate the proposed abstraction-based verification approach.
For future work, we plan to study the problem of controller synthesis to enforce pre-opacity for general control systems.


\bibliographystyle{plain} 
\bibliography{bibfile}

\end{document}